\documentclass[conference,a4paper]{IEEEtran}
\usepackage[T1]{fontenc}% optional T1 font encoding

\makeatletter
\def\endthebibliography{%
	\def\@noitemerr{\@latex@warning{Empty `thebibliography' environment}}%
	\endlist
}
\makeatother

\IEEEoverridecommandlockouts
\usepackage{cite}
\usepackage{xcolor}
\usepackage{amsmath,amssymb,amsfonts}
\usepackage{algorithm}
\usepackage{algorithmic}
\usepackage{etoolbox}  % patch def of algorithmic environment
\usepackage{bbm}
\usepackage{graphicx}
\usepackage{amsmath}
\usepackage{amssymb}
\usepackage{amsthm}
\usepackage{bm}
\usepackage{gensymb}
\usepackage{booktabs}
\usepackage{subfig}
\usepackage{epstopdf}
\DeclareMathOperator*{\argmin}{arg\,min\,}

\makeatletter
\patchcmd{\algorithmic}{\addtolength{\ALC@tlm}{\leftmargin} }{\addtolength{\ALC@tlm}{\leftmargin}}{}{}
\makeatother

\makeatletter
\newcommand\fs@betterruled{%
	\def\@fs@cfont{\bfseries}\let\@fs@capt\floatc@ruled
	\def\@fs@pre{\vspace*{5pt}\hrule height.8pt depth0pt \kern2pt}%
	\def\@fs@post{\kern2pt\hrule\relax}%
	\def\@fs@mid{\kern2pt\hrule\kern2pt}%
	\let\@fs@iftopcapt\iftrue}
\floatstyle{betterruled}
\restylefloat{algorithm}
\makeatother

\definecolor{BuGn}{RGB}{28,144,153}
\newcommand{\ilm}{}
\usepackage{filecontents}
\usepackage{graphicx}
\usepackage{textcomp}
\usepackage{xcolor}
\def\BibTeX{{\rm B\kern-.05em{\sc i\kern-.025em b}\kern-.08em
		T\kern-.1667em\lower.7ex\hbox{E}\kern-.125emX}}

\newtheorem{thm}{Theorem}

\newtheorem{prop}[thm]{Proposition}

\begin{document}
\bstctlcite{IEEEexample:BSTcontrol}
\title{Inter-plane satellite matching in dense LEO constellations}

\author{Beatriz~Soret, Israel~Leyva-Mayorga, \ilm{and} Petar Popovski\\
Department of Electronic Systems, Aalborg University, 9220, Aalborg, Denmark\\
Email:\{bsa, ilm, petarp\}@es.aau.dk
}

%{Shell \MakeLowercase{\textit{et al.}}: Bare Demo of IEEEtran.cls for IEEE Communications Society Journals}
\maketitle
\begin{abstract}
Dense constellations of Low Earth Orbit (LEO) small satellites are envisioned to make extensive use of the inter-satellite link (ISL). Within the same orbital plane, the inter-satellite distances are preserved and the links are rather stable. In contrast, the relative motion between planes makes the inter-plane ISL challenging. In a dense set-up, each spacecraft has several satellites in its coverage volume, but the time duration of each of these links is small and the maximum number of active connections is limited by the hardware. We analyze the matching problem of connecting satellites using the inter-plane ISL for unicast transmissions. We present and evaluate the performance of two solutions to the matching problem with any number of orbital planes and up to two transceivers: a heuristic solution with the aim of minimizing the total cost; and a Markovian solution to maintain the on-going connections as long as possible. The Markovian algorithm reduces the time needed to solve the matching up to $1000\times$
and $10\times$ 
with respect to the optimal solution
and to the heuristic solution, respectively, 
without compromising the total cost. %The last example assumes two transceivers for the inter-plane connectivity. 
Our model includes power adaptation and optimizes the network energy consumption as the exemplary cost in the evaluations, but any other QoS-oriented KPI can be used instead. \end{abstract}

%\begin{IEEEkeywords}
%LEO satellite constellation, Inter-Satelite link, assignment problem, inter-plane, throughput
%\end{IEEEkeywords}

\IEEEpeerreviewmaketitle

\section{Introduction}

%\cite{TR22.822}
Low Earth Orbit (LEO) dense constellations of small satellites, based on the CubeSat architecture, 
have become an attractive solution for Internet of Things (IoT) applications in 5G~\cite{TR38.913}.
The constellation is composed of hundreds of spacecrafts plus several ground stations, working all together as a relay communication network. The space segment is organized in several orbital planes that can be deployed at different inclinations and altitudes \cite{Walker1984} \cite{Walker1971}. 
The satellites are connected to each other via the Inter-Satellite Links (ISL), a two-way connection. The ISL can be intra-plane ISL, connecting with the satellite in front and the satellite behind in the same plane; and inter-plane ISL, connecting satellites from different orbital planes. In addition, the satellites are connected to ground stations, gateways or end-devices through the Ground-to-Satellite Link (GSL).
%, which is used for Telemetry and Telecommand (TMTC) data and device data. %(see Figure \ref{fig_links}). 
LEO satellites move at speeds $>25\,000$~km/h relative to the ground terminals. Therefore, the GSL is only available for a few minutes before handover to another satellite occurs. 

The use of the ISL unleashes the true potential of a LEO constellation, ensuring continuous connectivity, and reducing the number of required ground stations and the end-to-end latency. %Examples of current applications exploiting the ISL are TMTC, which sends telecommands to multiple LEO satellites, or retrieval of surveillance data from the constellation. In addition, 
One example of application is to use the constellation as a relay network, which can dramatically increase the coverage of machine-type communication (MTC) and IoT deployments in rural or remote areas, where the cellular and other relaying networks are out of range~\cite{TR38.913}.

Inter-satellite distances are usually preserved within a plane. However, inter-satellite distances between different planes are time-variant: longest when satellites are over the Equator, and shortest over the polar region boundaries. Moreover, the orbital periods are different if the planes are deployed at different altitudes, or if these contain a different number of satellites, which results in aperiodic topologies. In a dense set-up, each spacecraft has several inter-plane satellites in its coverage volume, which leads to a matching problem of who should communicate to whom. 

Although less investigated than the GSL, several works have addressed the communication challenges of the ISL. The authors in \cite{Radhakrishnan2016} provide a thorough compilation of the latest research efforts in the area of inter-satellite communications, organized in physical, data and network layer. \cite{Qu2017} describes the main use cases and elements of a LEO constellation for IoT, including the use of the ISL. In~\cite{Popescu2017}, a power budget analysis for CubeSats that includes the ISL is conducted. \cite{Cowley2009} addresses the communication among a group of independent satellites in an unstructured constellation, treating the spacecraft positions as random variables. 

Matching problems are among the most important problems in network optimization~\cite{Conforti2014}. For unmanned aerial vehicles (UAVs), \cite{Liu2019} investigates the assignment problem in a Flying Ad-Hoc Network composed of drones, formulating a dynamic matching game that uses the trajectory of the drones. In this paper, we address the matching problem of finding the inter-plane ISL connections that minimize the total cost of the constellation at each time instant. The model includes power adaptation and the power consumption is the exemplary cost, but any other QoS-oriented Key Performance Indicator (KPI) can be optimized instead. Differently than \cite{Cowley2009}, we address a planned network and solve the combinatorial problem by considering the predictability of the spacecrafts positions. 

%There are several factors to be considered in order to find the best matching. First of all, the link availability, which depends on the inter-satellite distances, the link budgets and the orientation of transmitter and receiver. Hence, it can be highly dynamic between orbits. Secondly, the pending traffic in each satellite buffer, together with the desired route. The link availability and the traffic load determine the set of potential connections. Finally, the time scale at which the ISL matching algorithm should work, given the time scale of the topology movement and the ISL communication. With all these elements, the best pairing between elements in the set is decided.
%This paper addresses the problem of satellite matching in the ISL of a dense LEO constellation. 

Specifically, we aim to solve the inter-plane matching problem with $M$ orbital planes and for up to two simultaneous ISLs per satellite. The Hungarian algorithm~\cite{Munkres1957} is known to find the optimal pairing in bipartite graphs, which corresponds to the case with only $M=2$ and one ISL. Furthermore, its computational cost is high. Conversely, we take a network-wise approach and propose two novel algorithms that provide a near-optimal solution to the matching problem with any $M$ and up to two simultaneous ISLs per satellite without compromising the total cost. In addition, the computational complexity of our two algorithms is  much lower than that of the Hungarian algorithm.
%We notice that the cost in resources of establishing a new inter-plane connection cannot be neglected, such that the current on-going connections should be maintained as long as possible. Moreover, this gives a simpler solution to the problem, since it is based on the previous assignments.  

The rest of the paper is organized as follows. In \mbox{Section \ref{sec:systemmodel}}, we describe the system model. Sections \ref{sec:1modem} and \ref{sec:2modems} address the problem when the CubeSat is equipped with one and two transceivers, respectively.  Section~\ref{sec:results} presents the performance results and  Section \ref{sec:conclusions} the conclusions. % on the complexity of our approaches, power adaptation, and the inclusion of a second inter-plane transceiver.  %Section \ref{sec:discussion} discusses some practical constraints. 
%Final remarks are given in Section \ref{sec:conclusions}. 

\section{System model} \label{sec:systemmodel}

\subsection{Geometry}
The constellation is composed of $N$ satellites distributed in $M$ circular orbital planes. Planes $m=1,2,\dotsc, M$ are composed of $N_m$ evenly distributed satellites, and each orbital plane is defined by the altitude $h_m$, the inclination $\epsilon_m$ and the orbital period $T_m$. Each of the $N = \sum_m N_m$ satellites in the constellation is assigned an index $i\in\{1, 2,\dotsc, N\}$ that serves as a unique identifier. $\mathcal{P}(i)$ is the set of satellites in the same orbital plane as $i$. The function ${\mathtt{p}(\cdot)}$ gives the plane of a satellite. If the number of satellites per plane is the same for all planes ($N_1 = N_2 = ...$), then 
\begin{equation}
\mathtt{p}(i)= \left \lfloor \frac{ {i-1}}{N_1}\right \rfloor + 1 
\end{equation}

Orbits with a low inclination are called equatorial or near equatorial orbits, and polar orbits are those passing above or nearly above both poles on each revolution (i.e., $\epsilon_m$ close to $\pi/2$). There are two classical topologies: the Walker star or polar \cite{Walker1984}, and the Walker $\delta$ or Rosette \cite{Walker1971} 
\cite{Ballard1980}. Without loss of generality, the results of this study are obtained for a Walker $\delta$ constellation like the one shown in \figurename~\ref{fig:constellation}, whose specific parameters are given in Section~\ref{sec:results}.

 \begin{figure}[t]
	\centering
	\includegraphics{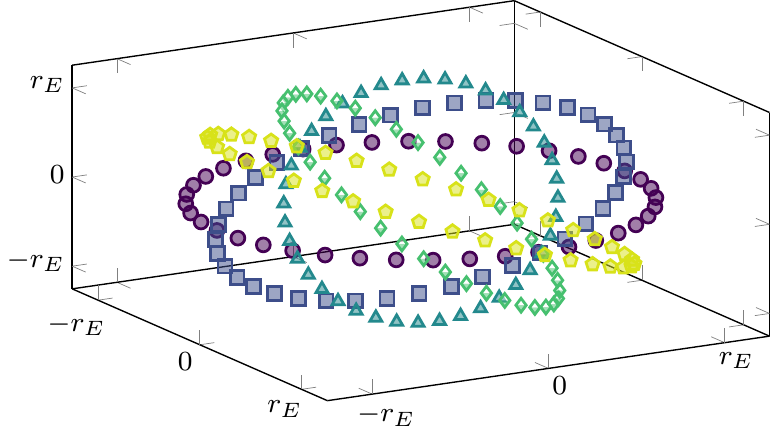}
	\caption{Walker $\delta$ constellation with $N=200$ satellites in $M=5$ orbital planes with altitudes starting at $1000$~km; the Earth radius is $r_E=6371$~km.}
	\label{fig:constellation}
\end{figure}

\subsection{Antennae placement}
The attitude determination and control subsystem of \mbox{CubeSats} is often specified to be 3-axis, stabilized with the yaw axis (x-axis) pointing towards the zenith, the z-axis (pitch) aligned to the orbit angular momentum (i.e., perpendicular to the orbit plane), and the y-axis (roll) aligned to the satellite velocity vector.

Although a set of coordinated small satellites have similar functionality as a \textit{big} satellite, there are practical constraints in the design of each CubeSat in terms of energy, weight and processing. 
%The design of CubeSats introduces practical constraints in terms of energy, weight, and processing power. 
Some of these constraints are related to the cube structure itself. For instance, the position of the antennas is rarely free due to the satellite geometry and the placement of other subsystems like thrusters, payload, and heat shielding. Furthermore, even when the inter-plane ISL is implemented, a practical mission will typically prioritize the stability of the GSL and the intra-plane ISL. Under these premises, the GSL antennas will be pointing towards the Earth's center, in the yaw axis, with a dedicated modem%; the payload is preferably installed on the opposite side
. The intra-plane ISL antennas are deployed in both sides of the roll axis, and two intra-plane transceivers are required to ensure two-way communication within an orbital plane. The pitch axis is then left for the inter-plane ISLs antennas and, depending on weight restrictions, one or two transceivers can be placed for this connectivity type. Both cases, one and two modems, are considered in this paper.

%Inter-plane, the minimum is to have only one modem, meaning that the satellite can keep only one inter-plane communication at a time. The full potential is obtained though using both sides of the CubeSat and two modems. In this paper we consider both cases: having one and two inter-plane transceivers. 

%\begin{figure}[t]
%	\centering
%	\includegraphics[width=2.5in]{rpy.png}
%	\caption{Roll, pitch, and yaw axes.}
%	\label{fig_rpy}
%\end{figure}

\subsection{Link budget and power adaptation} \label{sec:poweradap}
For the sake of notation simplicity, we skip the time dependence $t$ in the following. At any given time, the received SNR at satellite $j$ from satellite $i\neq j$ is written as
\begin{equation}
    \text{SNR}(i,j) = \frac{P_t G_t G_r}{k T_s R L_p(i,j)}
\end{equation}
where $P_t$ is the transmission power; $G_t$ and $G_r$ are the transmit and receive antenna gains, respectively; $k$ is Boltzmann's constant; $T_s$ is the system noise temperature; $R$ is the data rate in the radio link; and $L_p(i,j)$ is the free-space propagation path loss between satellites $i$ and $j$. The latter is given as
\begin{equation}
    L_p(i,j) = \left( \frac{4\pi l(i,j) f}{c} \right) ^2 \label{eq:pathloss}
\end{equation} 
where $l(i,j)$ is the line-of-sight distance (or slant range) between satellites $i$ and $j$, $f$ is the transmission frequency, and $c$ is the light speed.

\begin{prop}
The slant range between neighboring satellites $i$ and $j$ in orbital plane $a=\mathtt{p}(i)=\mathtt{p}(j)$ is given by
\begin{align}
    %BS: the old version was taking the arc distance, we replace by the LOS distance
    %l_\ell =\min \left\{l(i,j) \mid a=\mathcal{P}(i)=\mathcal{P}(j)\right\}= \frac{2 \pi (r_{E} + h_a)}{N_a}
    l_\text{intra}(a) &=\min \left\{l(i,j) \mid a=\mathtt{p}(i)=\mathtt{p}(j)\right\}\nonumber\\
    &= 2(r_E+h_a)\cos\frac{\pi}{N_a}\tan\frac{\pi}{N_a} 
    \label{eq:distance_intra}
    \end{align}
where $r_E$ is the radius of the Earth. 

The slant range between satellites $i$ and $j$ in orbital planes $a=\mathtt{p}(i)\neq b=\mathtt{p}(j)$, respectively, is given by
\begin{align}
 \label{eq:distance}
l(i,j)& = 
%\sqrt{(x_{i,n}-x_{j,m})^2 + (y_{i,n} - y_{j,m})^2 + (z_{i,n}-z_{j,m})^2  }
%= 
\left[ (h_a+r_E)^2 + (h_b+r_E)^2 \right.\nonumber\\
& \left. -2(h_a+r_E) (h_b+r_E)\cos\theta_{a,i}\cos\theta_{b,j} \right.\nonumber
\\ &\left. -2(h_a+r_E) (h_b+r_E)\cos(\epsilon_a-\epsilon_b)\sin\theta_{a,i}\sin\theta_{b,j} \right] ^{1/2} \; .
\end{align}
\end{prop}

\begin{proof}
Equation~\eqref{eq:distance_intra} is derived from a circular orbit with evenly distributed satellites, by calculating the distance between two points in a circle. To calculate the distance between spacecrafts in different orbital planes, as in~\eqref{eq:distance}, let $T_a$ denote the orbital period of plane $a$ and $\theta_{a,i}(t) = \left(2\pi t/T_a\right) + \left(2\pi i/N_a\right)$ denote the orbital angle of satellite $i$ in plane $a$ at time $t$. Notice the notation emphasis here regarding the time dependence. The equatorial coordinates of the satellite are first expressed in terms of the ecliptic coordinates in the ecliptic plane, $(\text{x}_{ecl}, \text{y}_{ecl}, \text{z}_{ecl})$, with the x-axis aligned toward the equinox. After some calculations the equatorial coordinates are written
\begin{IEEEeqnarray}{rCl}
    \text{x}_{a,i}(t) &=& (h_a+r_E) \cos\theta_{a,i}(t) \IEEEyesnumber\IEEEyessubnumber\\
    \text{y}_{a,i}(t) &=& (h_a+r_E) \cos\epsilon_i \sin \theta_{a,i}(t)\IEEEyessubnumber \\
    \text{z}_{a,i}(t) &=& (h_a+r_E) \sin\epsilon_i \sin \theta_{a,i}(t)\IEEEyessubnumber
\end{IEEEeqnarray}
Then, the euclidean distance is calculated to obtain (\ref{eq:distance}).
\end{proof}

The achievable ISL data rate is constrained by the usual link budget parameters including, among others: modulation and coding schemes, link distance, frequency, RF power, antenna gains, noise temperature, and equipment losses. We study a scenario in which the spacecrafts aim to transmit at a minimum data rate $R$ in the ISL. We assume that the link budget parameters listed above remain constant, except for $P_t$. Mathematically, \ilm{the minimum $P_t$ needed to achieve $R$ is}
\begin{align}
    %P_t&=
    &\min_{P_t} \left\{P_t\mid R{\leq} B\log_2 \left(1+\text{SNR}(i,j)\right)\right\}\nonumber\\
    &=\frac{\left(2^{R/B}-1\right) k T_s R}{ G_t G_r} \left(\frac{4 \pi  {l(i,j)} f}{c}\right)^2. 
\end{align}

The power adaptation policy together with the constellation geometry determines the link opportunities, or contact times, defined as the time a pair of satellites are within the communication range. Although not considered in this paper, the contact time is also influenced by the directivity and the radiation pattern of the antennae (i.e., if both transmitting and receiving satellite are pointing to each other), but we assume omnidirectional antennae throughout the paper and leave the antenna characterization for future work.  

As power adaptation policy, we define two levels: low power $P_\ell$ and high power $P_h$. The maximum distances at which two satellites $i$ and $j$ can communicate at the low power and high power level at a minimum rate $R$ are design parameters, given by the hardware and energy constraints of the spatial mission. Through this study, $P_\ell$ is set to be the minimum power to achieve a theoretical channel capacity that is higher than $R$ with a bandwidth $B$ and within a distance \ilm{\mbox{$l_\ell = \eta \cdot l_\text{intra}(m)$}, where $\eta$ is the design parameter.
%, with $l_\text{intra}(m)$ given by equation (\ref{eq:distance_intra}). 
$P_h$ is defined analogously but within a distance $l_{h} = 2 \eta \cdot l_\text{intra}(m)=2 l_\ell$.}
%The latter, in combination with the constellation geometry, determines the link opportunities, defined as the time a pair of satellites are within the communication range. 
%It is straightforward to extend our study to the case with any number of power levels. 
The inter-plane ISL transmission power is then selected as
%Throughout this study, the transmission power used for an inter-plane ISL is selected as
\begin{equation}
    P_t=\begin{cases}
    P_\ell, & \text{if } l(i,j)\leq l_\ell\\
    P_h, & \text{if }l_\ell<l(i,j)\leq l_{h}.
    \end{cases}
\end{equation}
%The impact on performance of different values for $l_\ell$ and $l_h$ is explored in Section~\ref{sec:results}. 

\section{Matching problem with one transceiver} \label{sec:1modem}
In this section we define the inter-plane ISL matching problem with one transceiver and propose two approaches to solve it. Then, the matching problem and our approaches are extended to the case with two transceivers in Section~\ref{sec:2modems}. 

The satellite network can be represented as a time-varying graph in which the inter-plane ISL link opportunities are short. The dynamics is such that satellites may perform early handover (i.e., before reaching the maximum distance $l_h$) to increase their link budget and, hence, transmit at a higher data rate and/or at a lower power when compared to the previous inter-plane link. 
Therefore, the actual link duration (or contact time) between two satellites may be shorter than the link opportunity. However, handover has an inherent signaling, delay, and processing overhead. Consequently, excessively frequent handovers must be avoided, for which a minimum period between handovers $\ilm{T_\text{ho}}$ can be selected. Then, the matching problem can be solved once every $\ilm{T_\text{ho}}$ by taking samples (snapshots) of the constellation. Hence, $\ilm{T_\text{ho}}$ is denoted as the sampling period. Finding an optimal value for $\ilm{T_\text{ho}}$ is out of the scope of this paper. Instead, we 
set a sufficiently short $\ilm{T_\text{ho}}$ to consider a static geometry of the constellation during this period. 
%focus on 
%the solution of the matching problem (i.e., on minimizing the cost of the inter-plane ISL) for a constellation at each time instant $t\in \{k\,T_{s}:k\in \mathbb{N}\}$, 
%set for a sufficiently short $T_{s}$ to consider a static geometry of the constellation during this period.  

Inter-plane antennas are placed in the $Y+$ and $Y-$ sides of the spacecraft, so the coverage volume with one transceiver is assumed to be the same as the implementation with two transceivers, but only one simultaneous inter-plane connection is possible. For $M\in\mathbb{Z}_+$ orbital planes, this is the matching problem in a $M$-partite graph described in the following. 

Let $G(V,E)$ be a graph with set of vertices $V$ and set of edges $E$. Graph $G(V,E)$ is $M$-partite if $V$ can be divided into $M$ disjoint subsets $V_m$; hence, $V=\bigcup_{m=1}^M V_m$, \mbox{$\bigcap_{m=1}^M V_m=\emptyset$}. Each subset $V_m$ represents one orbital plane, so $|V_m|=N_m$ and each edge in $E$ has endpoints $(i,j)$, given $\{i,j\}\in V$ and $j\notin \mathcal{P}(i)$. These endpoints are usually called \textit{agents} and \textit{tasks}. Any agent can be assigned to perform any task, incurring some cost that may vary depending on the agent-task assignment. The cost of using edge $(i,j)\in E$ is denoted as $w_{ij}$.

At any given $s$, $w_{ij}$ is a function of the the power level $P_{ij}$ required for communication at the edge $(i,j)$ between satellites $i\in V_a$ and $j\in V_b$ (i.e., from planes $a$ and $b$, respectively). We write
\begin{equation}
w_{ij} = \frac{P_{ij}}{ \mathbbm{1}(Q>0)} 
\end{equation}
\noindent where $\mathbbm{1}(Q>0)=1$ if the transmission buffer is not empty, i.e., $Q > 0$ and $\mathbbm{1}(Q>0)=0$ otherwise. 

Let $W$ be the symmetric matrix with all costs $w_{ij}$. Matrix $W$ is formed by block matrices $W_{a,b}$, which contain the costs $w_{ij}$ for all $i\in V_a$ and $j\in V_b$. Naturally, $(i,j)\notin E$ if $j\in\mathcal{P}(i)$; hence, we set $w_{ij}=\infty$ in these cases, which gives
\[
W=\begin{bmatrix}
\infty & W_{1,2} & \cdots & W_{1,M}\\
W_{2,1} & \infty & \cdots & W_{2,M}\\
\vdots & \vdots & \ddots & \vdots\\
W_{M,1} & W_{M,2} & \cdots & \infty\\
\end{bmatrix}
\]

The goal is to assign exactly one agent to each task and exactly one task to each agent in such a way that the total cost of the assignment is minimized. Let $A\subseteq E$ be an assignment  on $G(V,E)$. The assignment $A$ that results in the minimum cost among all the possible assignments is optimal. 
The matching problem is mathematically written
\begin{equation}
\begin{aligned}
& \underset{}{\text{min}} & & \sum_{i=1}^N\sum_{\substack{j=1 \\ j \notin \mathcal{P}(i)}}^N w_{ij}x_{ij} \\
& \text{subject to } & & \sum_{\substack{j=1 \\ j \notin \mathcal{P}(i)}}^N x_{ij}= 1 \;\; \forall i \\
& & & x_{ij} \in \{0,1\} 
\end{aligned} \label{eq:mplanes}
\end{equation}

\noindent where $x_{ij}=1$ indicates a match (i.e., $(i,j)\in A$) and $x_{ij}=0$ indicates no match between spacecrafts.

\subsection{Independent experiments matching}
\label{sec:indep_experiments}

This is the classical static approach, in which the underlying graph is assumed to be time-invariant. Hence, the matching problem is solved at each time instant $t$ without taking into consideration past decisions. Therefore, this solution minimizes the immediate cost at each $t$ independently.

The basic Hungarian algorithm \cite{Munkres1957} gives the optimal solution in polynomial time ($\mathcal{O}(N^3)$) for the independent experiments matching problem in~\eqref{eq:mplanes} for $M=2$. Nevertheless, the asymptotic complexity of the Hungarian algorithm is restrictive in constellations where the number of satellites is large and extensions are needed for the cases where $M>2$. In these cases, we propose the use of the heuristic Algorithm~\ref{alg:suboptimal}, which greatly reduces the computational complexity with respect to the Hungarian algorithm and provides a near-optimal solution for the independent experiments matching. Its operation is summarized as follows. At each time instant $t$, the weights $w_{ij}$ are updated, and the strategy is to recursively add edges to the set of assignments $A$ by finding the edge $(i,j)$ with the smallest weight. Then, the rows and columns with the indices of the new pair are deleted from $W$. 

%The Hungarian algorithm is based on the following theorem.

%\newtheorem{theorem}{Theorem}
%\begin{theorem}
%If any number is added to or subtracted from all of the entries of any one row or column of a cost matrix, then an optimal assignment for the resulting cost matrix is also an optimal assignment for the original cost matrix. 
%\end{theorem}

%\textcolor{red}{NOTA: Me he cargado el siguiente p\'arrafo, no nos ayuda en nada}
%It has been verified by simulations with a variety of constellation geometries that this near-optimal algorithm provides a closely similar solution to that with the Hungarian algorithm. The computational complexity of this near-optimal algorithm is compared to that of the Hungarian algorithm in Section~\ref{sec:results} in terms of execution time.% on page~\pageref{fig:execution_time}.

\begin{algorithm} [t]
	\centering
	\caption{Heuristic algorithm for independent experiments matching with a single transceiver.}
	\begin{algorithmic}[1] 
	\renewcommand{\algorithmicrequire}{\textbf{Input:}}
		\renewcommand{\algorithmicensure}{\textbf{Output:}}
		\REQUIRE $W$ is matrix of costs
		\STATE $x_{ij}=0$ for all $i, j$ 
		\WHILE { $\exists \min W<\infty$}
		\STATE $i^*,j^*  \longleftarrow \underset{i,j}{\argmin} W  $
		\STATE $x_{i^* j^*}=1$ (associate satellite $i^*$ to satellite $j^*$)
		\STATE Delete rows and columns with indices $i^*$ and $j^*$
		\ENDWHILE
		
	\end{algorithmic}  \label{alg:suboptimal}
\end{algorithm}

\subsection{Markovian matching: maximization of the contact time}
\label{sec:markovian_matching}
The weights $w_{ij}$ change with the movement of the constellation, which is predictable. Given a sufficiently short $\ilm{T_\text{ho}}$, the movement of the constellation from $t$ to $t+\ilm{T_\text{ho}}$ is smooth and relatively slow. Therefore, it is likely that most of the pairs assigned at $t$ are still near-optimal choices at $t+\ilm{T_\text{ho}}$. Consequently, the slow and predictable time evolution of the geometry between these time instants can be exploited to increase the contact time (and reduce the handovers) and reduce the computational complexity of the matching algorithm. For this, we formulate a dynamic assignment problem in which the previous state of the system is considered. 

Let $\left\{X_s^{(i,j)}\right\}_{s}$ be the stochastic process with time index $s$ and state space $\{0,1\}$ that defines the inter-plane matching between satellites $i$ and $j$ at time index $s$. That is, $\Pr\left[X_s^{(i,j)}=1\right]=\Pr\left[(i,j)\in A\right]$ and vice versa. We denote $x_{i j}(s)$ as the event of a match between $i$ and $j$ at time index $s$; hence, $x_{i j}(s)$ is analogous to $x_{i j}$ for experiment $s$.

Algorithm~\ref{alg:suboptimal} is extended to maintain existing satellite pairs for as long as $l(i,j)\leq l_{h}$. For this, we define
\begin{equation}
    \Pr\left[X_s^{(i,j)}=1\mid X_{s-1}^{(i,j)}=1, w_{ij}<\infty\right]=1;
\end{equation}
hence, these pairs are eliminated from $W$. Only then, new satellite pairs are created according to Algorithm~\ref{alg:suboptimal}. Therefore, $\left\{X_s^{(i,j)}\right\}_{s}$ is a discrete-time Markov chain.
Algorithm~\ref{alg:suboptimallinkduration} summarizes the Markovian approach, whose computational complexity is expected to be considerably lower than that of the independent experiments matching.

%Figure \ref{fig:indep_vs_markovian} shows the performance comparison in terms of link duration CDF for the constellation in Figure \ref{fig:constellation}. 

\begin{algorithm} [t]
	\centering
	\caption{Markovian algorithm for a single transceiver}
	\begin{algorithmic}[1] 
			\renewcommand{\algorithmicrequire}{\textbf{Input:}}
		\renewcommand{\algorithmicensure}{\textbf{Output:}}
		\REQUIRE $W$ is the matrix of costs at time index $s$
		\REQUIRE $x_{ij}(s-1)$ are the matching indicators at $s-1$, 
		\STATE $x_{ij}(s)=0$ for all $i,j$ 
        \FOR {$x_{ij}(s-1)=1$}
            \IF {$w_{i j} < \infty $ }
            \STATE {\{\textit{The pair of satellites is still reachable}\}}
                \STATE $x_{i j}(s)=1$ (maintain association)
                \STATE Delete rows and columns with indices $i$ and $j$
            \ENDIF
        \ENDFOR
		\WHILE {$\exists \min W<\infty$}
		\STATE Find $i^*,j^*  \longleftarrow \underset{i,j}{\argmin} W$
		\STATE $x_{i^*j^*}(s)=1$ (associate satellite $i^*$ to satellite $j^*$) 
		\STATE Delete rows and columns with indices $i^*$ and $j^*$
		\ENDWHILE
	\end{algorithmic}  \label{alg:suboptimallinkduration}
\end{algorithm}

%\begin{figure}[t]
%	\centering
%	\includegraphics[width=3.8in]{link_duration_results.png}
%	\caption{Comparison of the link duration in the independent matching solution vs. the Markovian solution for a single transceiver. The constellation is the one in Figure \ref{fig:constellation}. }
%	\label{fig:indep_vs_markovian}
%\end{figure}

\section{Matching problem with two transceivers} \label{sec:2modems}
In this section, we extend the algorithm and the formulation for the independent experiments matching introduced in Section~\ref{sec:indep_experiments} to the case in which each satellite is equipped with two inter-plane transceivers. The extension of the following to the Markovian approach described in Section~\ref{sec:markovian_matching} is straightforward.

For each satellite $i$, its orbital plane defines two possibilities for the relative position with respect to satellite $j\notin\mathcal{P}(i)$. That is, $j$ is either to the $Y+$ side or to the $Y-$ side of $i$. Since the inter-plane antennas are placed in opposite sides of the satellite, at most one ISL can be maintained at each side $Y+$ and $Y-$. Building on this, the matching problem becomes
\begin{equation}
\begin{aligned}
& \underset{}{\text{min}} & &\sum_{i=1}^{N}  \sum\limits_{\substack{j=1 \\ j \notin \mathcal{P}(i)}}^{N}w_{ij}x_{ij}^{Y+} +w_{ij}x_{ij}^{Y-}\\
& \text{subject to} & &\sum\limits_{\substack{j=1 \\  j \notin \mathcal{P}(i)}}^{N}x_{ij}^{Y+} = 1 \;\; \forall i \\
& & &\sum\limits_{\substack{j=1 \\ j \notin \mathcal{P}(i)}}^{N}x_{ij}^{Y-} = 1 \;\; \forall i \\
& & &x_{ij}^d \in \{0,1\}, \; d \in\{Y+,Y-\}
\end{aligned} \label{eq:2modems}
\end{equation}

Algorithm~\ref{alg:2modems} solves the problem in equation~\eqref{eq:2modems}, using the same principles introduced for the case with one transceiver. The extension to two transceivers is possible by allowing one matching at each $Y+$ and $Y-$ (indicated by $x_{ij}^d$, where $d\in\{Y+,Y-\}$). Satellite $i$ is removed from $W$ only when a matching is made at each side. 

\begin{algorithm} [t]
	\centering
	\caption{Algorithm for two-transceivers}
	\begin{algorithmic}[1] 
			\renewcommand{\algorithmicrequire}{\textbf{Input:}}
		\renewcommand{\algorithmicensure}{\textbf{Output:}}
		\REQUIRE $W$ is the matrix of costs

		\WHILE {$\exists\min W < \infty$}
		\STATE Find $i^*,j^*  \longleftarrow \underset{i,j}{\argmin} W$
        \STATE Find 
        $d$ for $j$ with respect to $i$ and vice versa
		\IF {$x_{i^*j^*}^{d} == 0 \;\&\&\; x_{j^*i^*}^{d} == 0$}
    		\STATE $x_{i^*j^*}^{d}=1$ and $x_{j^*i^*}^{d}=1$
            \STATE $w_{ij} = \infty$
    		\IF {$\sum_k x_{i^*k}^{Y+} + x_{i^*k}^{Y-} == 2$}
    		    \STATE Delete the row and column with index $i^*$
    		\ENDIF
    		\IF {$\sum_k x_{j^*k}^{Y+} + x_{j^*k}^{Y-} == 2$}
    		    \STATE Delete the row and column with index $j^*$
    		\ENDIF
		\ENDIF
		\ENDWHILE
		
	\end{algorithmic}  \label{alg:2modems}
\end{algorithm}

\section{Results}
\label{sec:results}
This section presents the most relevant results derived from our analysis. % of the efficacy and computational complexity of our matching approaches. 
A Walker $\delta$ constellation, such as the one illustrated in Fig.~\ref{fig:constellation}, and the model from Section~\ref{sec:systemmodel} are considered. The default parameters involved in the geometry, link budget and power adaptation are as follows. A total of $N$ satellites are distributed in $M$ orbital planes deployed at heights $h_m = 900 + 100m$~km with $N_m=N/M$ for $m = 1,2,\dotsc,M$. The inclination of plane $m$ is $\epsilon_m = 2 \pi(m-1)/M$. 
The intra-plane distance at the highest orbital plane is used for power adaptation. Unless otherwise specified, $\eta = 1$, such that $l_\ell=l_\text{intra}(\ilm{m})$ and $l_{h} = 2 l_\text{intra}(\ilm{m})$. %The sampling period is $T_s=1$~s, which is much shorter than the minimum orbital period $T_1 = 6298$~s given for the lowest orbital plane at $h_1=1000$~km; hence, it is safe to consider the constellation is static throughout a single sampling period. 
The buffer of the satellites is never empty% (a.k.a. full buffer traffic)
; hence, $\mathbbm{1}(Q>0)=1$ is constant and all pairs in the coverage volume are considered for the matching.

Simulators for the studied algorithms have been developed in \ilm{Python 3. Simulations were run on a PC with Ubuntu 18.04.2 LTS ($64$~bit), an Intel Core i7-7820HQ CPU, $2.9$~GHz, and $16$~GB RAM. The clock precision of this platform is $10^{-7}$~s.} %MATLAB 2018b. Simulations were run on a PC with MS Windows 10 ($64$~bit), an Intel Core i5-6200U processor, $2.3$~GHz, and $8$~GB RAM. 
No other processes with a relevant CPU usage were run during the execution of our code. 

The constellation is simulated for at least five orbital periods of the lowest orbital plane ($h=1000$~km) and, at least, $10\,000$ matching problems are solved. \ilm{Given the  orbital period for the lowest orbital plane at $h_1=1000$~km is $T_1 = 6298$~s, the sampling period is $\ilm{T_\text{ho}}=5 T_1/ 10\,000=3.41$~s; hence, it is safe to consider the constellation is static throughout a single sampling period.}

%We use a full buffer traffic model, such that all satellites always have pending packets in the queues. 
Given the uniformity of the constellation, the differences in terms of contact time between the heuristic algorithm for the independent experiments and the Markovian algorithm are insignificant, with a slight improvement by using the Markovian one. The differences in total cost (i.e., network energy consumption) are also minor. Where the algorithms differ is in the execution time per matching, as plotted in \figurename~\ref{fig:execution_time} for the case with one transceiver and $N_m=40$. In \figurename~\ref{fig:execution_time} (a), for $M=2$, the execution time of the optimal Hungarian algorithm is also included. For $M=5$, plotted in  \figurename~\ref{fig:execution_time} (b), only the heuristic (independent experiments) and the Markovian solution are compared. The Markovian algorithm reduces the time needed to solve the matching up to $1000\times$ and $10\times$ 
with respect to the optimal solution,
and to the heuristic solution, respectively.
%The reduction in execution time is significant when using the Markovian: the Markovian solution is executed, on average, $3.8\times$ and $176\times$ faster than the independent experiments matching and the Hungarian algorithm, respectively. 
Other uniform geometries have been simulated with similar conclusions. %Such a sharp reduction in the execution time is highly relevant for the network operation, since the constellation must:(1) solve the matching problem; (2) perform handover; and (3) transmit data within $T_{s}$. 
%Hence, by reducing the time it takes to solve the matching problem, a greater fraction of the time is available for sending the data.

\begin{figure}[t]
	\centering
	\subfloat[]{\includegraphics{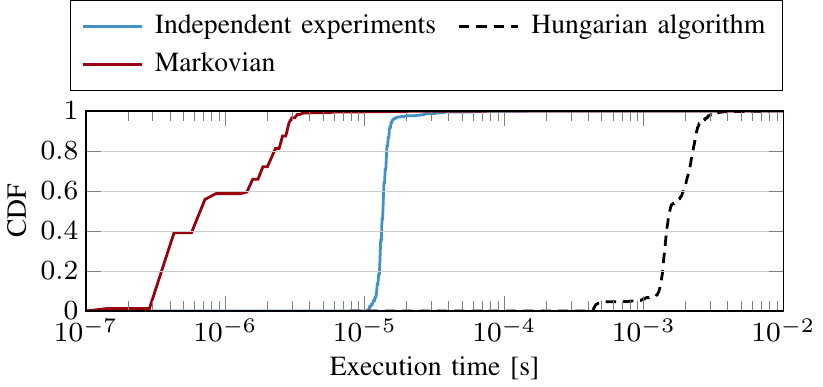}}\\
	\subfloat[]{\includegraphics{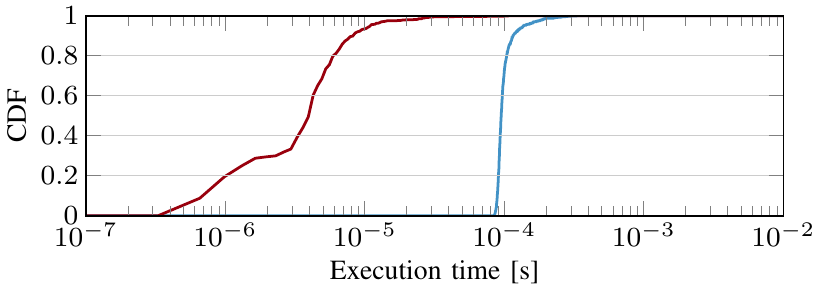}}
	\caption{Comparison of the execution time per matching with (a) $M=2$ and (b) $M=5$ using the Hungarian algorithm (only for $M=2$), the independent experiments, and the Markovian approaches for a single transceiver; $N_m=40$ for all $m$.}
	\label{fig:execution_time}
\end{figure}

\figurename~\ref{fig:avg_power} shows the relative average power consumption \mbox{$\bar{P} = \bar{P}_t/P_\ell$} when applying Algorithm~\ref{alg:2modems} to problem (\ref{eq:2modems}) in a constellation of $M=4,5,\dotsc,8$ orbital planes with $N_m=40$ for all $m$. The power is averaged over time and over the number of pairs (established inter-plane ISLs). The cases with $\eta = 1$ and $\eta = 2$ are plotted for the independent experiments and the Markovian algorithms. 
%The value of the latter depends on the distribution of the inter-plane distances for each geometry. For example, the blue line at the bottom of \figurename~\ref{fig:avg_power} corresponds to the case where $P_{\ell}$ is sufficient to establish $50\%$ of all possible inter-plane connections within $l_{h}=2 l_\text{intra}(M)$. 
As expected, the relative average power consumption is higher when $\eta = 2$, and consequently more pairs are candidates to establish a communication link. The Markovian solution results in a higher relative power as compared to the independent experiments, showing the price to pay in energy consumption for reducing the matching complexity and increasing the link duration, therefore reducing the handover signaling. 
%As expected, the average power increases if 
%the threshold is lower, and decreases as the number of planes increases, due to the lower inter-plane distances. Moreover, the low power links are prioritized, making the average power to remain very close to the lower bound for most of the scenarios. The variations observed across the x-axis of \figurename~\ref{fig:avg_power} are mainly due to the change of parity of $M$, which affects the number of potential ISLs.

\begin{figure}[t]
	\centering
	\includegraphics{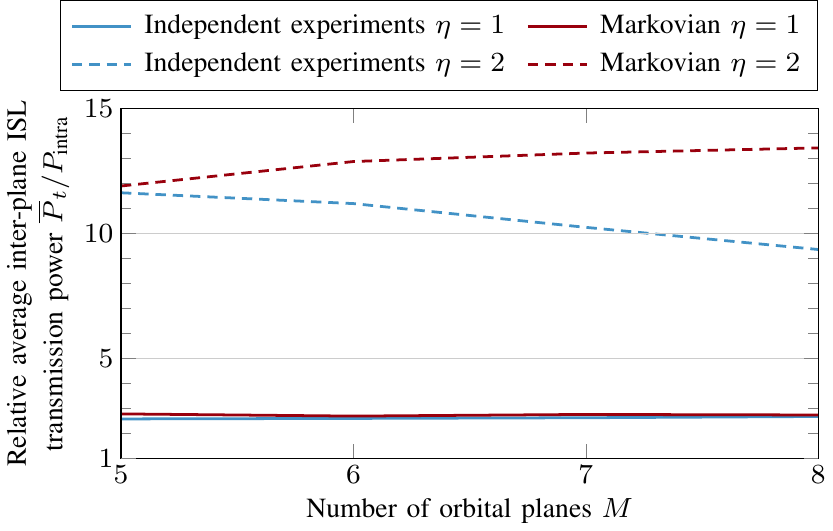}
	\caption{Relative average inter-plane ISL transmission power $\overline{P}_t/P_\text{intra}$ per satellite pair versus number of planes with the independent experiments and Markovian approaches; single transceiver.}
	\label{fig:avg_power}
\end{figure}

\figurename~\ref{fig:1vs2} shows the average number of satellite pairs established with $N_m=\{40,60\}$ and with one and two transceivers. As expected, the number of pairs increases with the $M$, $N_m$, and the number of transceivers. Nevertheless, the number of pairs with two transceivers is less than twice the number of pairs with one transceiver; this is the upper bound introduced by the geometry of the constellation. Therefore, the throughput in the inter-plane ISL of a constellation is not doubled by adding a second transceiver.

\section{Conclusions} \label{sec:conclusions}

We have addressed the inter-plane ISL in a LEO constellation of satellites using unicast communication. The constellation is modelled like a dynamic graph, in which vertices are satellites and edges are the communication links. The case in which the CubeSat is equipped with a single transceiver for this connectivity type is first studied, with a heuristic algorithm and a Markovian solution, the latter for maximizing the link duration. Then, the case with two transceivers is analyzed, considering the relative position of the planes. The cost of assigning a pair of spacecrafts is abstracted in our model, although the examples illustrate the minimization of the network energy consumption under a power adaptation scheme. The simulation results show that the Markovian solution sharply reduces the computational complexity of the matching when compared to the baseline algorithm. The algorithms are periodically executed with a sampling period sufficiently small, so it is straightforward to address changes in the topology due to malfunctioning spacecraft. 

\begin{figure}[t]
	\centering
	\includegraphics{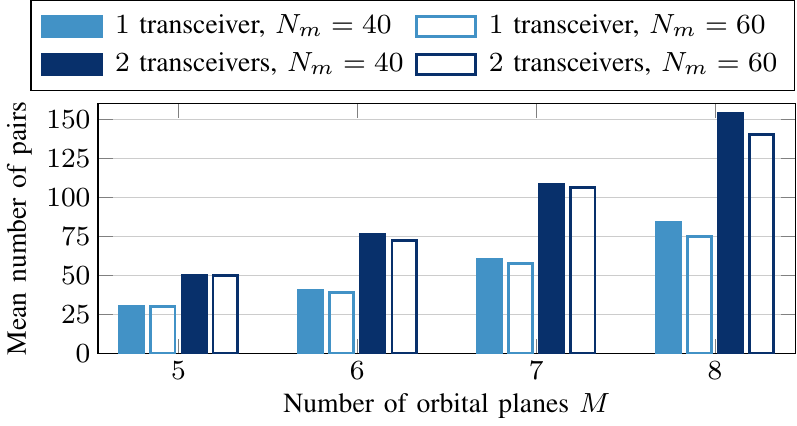}
	\caption{Mean number of satellite pairs versus number of planes and satellites per plane with one and two transceivers. }
	\label{fig:1vs2}
\end{figure}

This paper provides benchmark results for understanding the limits in the inter-plane ISL connectivity. The network-wise solution provided by our algorithms can be the basis to define practical implementations of distributed protocols that can be autonomously computed in each satellite. Another extension is the characterization of directional antennas and the pointing in the inter-plane ISL. %, rather than omnidirectional antennas

%When all satellites have non-empty queues, the network-wise solution provided by our algorithms can be autonomously computed in each satellite. For the general case of variable load, a practical implementation requires a distributed protocol that takes into account the stochastic behaviour of the satellites' buffer. This is left for future work. 

%The design of such a distributed protocol is something we will address in our future work. 

%Further work will address this issue, along with the medium access control to mitigate interference between ISLs. 

\section*{Acknowledgment}

This work has been in part supported by the European Research Council (Horizon 2020 ERC Consolidator Grant Nr. 648382 WILLOW). 

\bibliographystyle{IEEEtran}
\bibliography{ISLassignment}
\iffalse
 \clearpage
 \includegraphics[width=\columnwidth]{constellation.png}
\includegraphics[width=\columnwidth]{execution_time.png}
\includegraphics[width=\columnwidth]{avg_power.png}
\includegraphics[width=\columnwidth]{1vs2.png}
\fi
\end{document}